\documentclass[a4paper,UKenglish,cleveref, autoref, thm-restate]{lipics-v2021}
\pdfoutput=1 

\hideLIPIcs  

\usepackage[ruled,vlined,linesnumbered,noend]{algorithm2e}
\SetKwInOut{Parameter}{Parameters}

\usepackage{framed}
\usepackage[framemethod=tikz]{mdframed}

\newenvironment{wrapper}[1]
{
	\begin{center}
		\begin{minipage}{\linewidth}
			\begin{mdframed}[hidealllines=true, backgroundcolor=gray!20, leftmargin=0cm,innerleftmargin=0.4cm,innerrightmargin=0.4cm,innertopmargin=0.4cm,innerbottommargin=0.4cm,roundcorner=10pt]
				#1}
			{\end{mdframed}
		\end{minipage}
	\end{center}
}

\title{Arboricity-Dependent Algorithms for Edge Coloring}

\author{Sayan Bhattacharya}{University of Warwick, United Kingdom}{s.bhattacharya@warwick.ac.uk}{}{}

\author{Mart\'{i}n Costa}{University of Warwick, United Kingdom}{martin.costa@warwick.ac.uk}{}{}

\author{Nadav Panski}{Tel Aviv University, Israel}{nadavpanski@mail.tau.ac.il}{}{}

\author{Shay Solomon}{Tel Aviv University, Israel}{shayso@tauex.tau.ac.il}{}{}

\authorrunning{S. Bhattacharya, M. Costa, N. Panski and S. Solomon} 

\Copyright{Sayan Bhattacharya, Mart\'{i}n Costa, Nadav Panski and Shay Solomon} 

\ccsdesc[500]{Theory of computation~Dynamic graph algorithms} 

\keywords{Dynamic Algorithms, Graph Algorithms, Edge Coloring, Arboricity} 

\category{} 

\relatedversion{} 

\acknowledgements{}

\nolinenumbers 

\EventEditors{}
\EventNoEds{0}
\EventLongTitle{}
\EventShortTitle{}
\EventAcronym{}
\EventYear{}
\EventDate{}
\EventLocation{}
\EventLogo{}
\SeriesVolume{}
\ArticleNo{}

\begin{document}

\maketitle

\begin{abstract}
    The problem of edge coloring has been extensively studied over the years. Recently, this problem has received significant attention in the \emph{dynamic setting}, where we are given a dynamic graph evolving via a sequence of edge insertions and deletions and our objective is to maintain an edge coloring of the graph.

    Currently, it is not known whether it is possible to maintain a $(\Delta+ O(\Delta^{1 - \mu}))$-edge coloring in $\tilde{O}(1)$ update time, for any constant $\mu > 0$, where $\Delta$ is the maximum degree of the graph.\footnote{We use $\tilde O(\cdot)$ to hide polylogarthmic factors.}
    In this paper, we show how to efficiently maintain a $(\Delta +O(\alpha))$-edge coloring in $\tilde O(1)$ amortized update time, where $\alpha$ is the arboricty of the graph. Thus, we answer this question in the affirmative for graphs of sufficiently small arboricity.
    
\end{abstract}



\newcommand{\C}{\mathcal{C}}

\section{Introduction}

Consider any graph $G = (V, E)$, with $n = |V|$ nodes and $m = |E|$ edges, and any integer $\lambda \geq 1$. A {\em (proper) $\lambda$-(edge) coloring} $\chi : E \rightarrow  [\lambda]$ of $G$ assigns a color $\chi(e) \in [\lambda]$ to each edge $e \in E$, in such a way that no two adjacent edges receive the same color. Our goal is to get a proper $\lambda$-coloring of $G$, for as small a value of $\lambda$ as possible. It is easy to verify that any such coloring requires at least $\Delta$ colors, where $\Delta$ is the maximum degree of $G$. On the other hand, a textbook theorem by Vizing~\cite{Vizing} guarantees the existence of a proper $(\Delta+1)$-coloring in any input graph.

This work focuses on the edge coloring problem in the \emph{dynamic setting}, where an extensive body of work has been devoted to this problem. Before describing our contributions, we first summarize the relevant state-of-the-art in the dynamic setting. 

\medskip
 \noindent {\bf Dynamic Edge Coloring.} In the dynamic setting, the input graph $G$ undergoes a sequence of updates (edge insertions/deletions), and throughout this sequence the concerned algorithm has to maintain a proper coloring of $G$. We wish to design a dynamic algorithm whose {\em update time} (time taken to process an update) is as small as possible. The edge coloring problem has received significant attention within the dynamic algorithms community in recent years. It is known how to maintain a $(2\Delta-1)$-coloring in $O(\log \Delta)$ update time~\cite{DBLP:conf/iccS/BarenboimM17,BhattacharyaCHN18}, and~Duan et al.~\cite{DuanHZ19} showed how to maintain a $(1+\epsilon)\Delta$-coloring in $O(\log^8 n/\epsilon^4)$ update time when $\Delta = \Omega(\log^2 n/\epsilon^2)$. Subsequently, Christiansen~\cite{Christiansen22} presented a dynamic algorithm for $(1+\epsilon)\Delta$-coloring with $O(\log^9 n \log^6 \Delta/\epsilon^6)$ update time, without any restriction on $\Delta$.
 More recently, Bhattachrya et al.~\cite{bhattacharya2024nibbling} showed how to maintain a $(1+\epsilon)\Delta$-coloring in $O(\log^4(1/\epsilon)/\epsilon^9)$ update time when $\Delta \geq (\log n/\epsilon)^{\Theta((1/\epsilon)\log(1/\epsilon))}$.
 At present, no dynamic edge coloring algorithm is known with a sublinear in $\Delta$ {\em additive approximation} and with $\tilde{O}(1)$ update time. We summarize the following basic question that arises.

\begin{wrapper}
    Is there a dynamic algorithm for maintaining a $(\Delta+ O(\Delta^{1 - \mu}))$-edge coloring with $\tilde{O}(1)$ update time, for any constant $\mu > 0$?
\end{wrapper}

\subsection{Our Contribution} We address the above question for the family of {\em bounded arboricity} graphs. Formally, a  graph $G = (V, E)$ has {\em arboricity} (at most) $\alpha$ iff: $$\left \lceil \frac{|E(G[S])|}{(|S| - 1)}\right \rceil  \leq \alpha \text{ for every subset } S \subseteq V \text{ of size } |S| \geq 2,$$ where $G[S]$ denotes the subgraph of $G$ induced by $S$ and $E(G[S])$ denotes the edge-set of $G[S]$. It is easily verified that the arboricity of any graph is upper bounded by its maximum degree. There are many instances of graphs, however, with very high maximum degree but low arboricity.\footnote{Think of a star graph on $n$ nodes. It has $\Delta = n-1$ but $\alpha = 1$.} Intuitively, a graph with low arboricity is {\em sparse everywhere}. Every graph excluding a fixed minor has $O(1)$ arboricity, thus the family of constant arboricity graphs contains bounded treewidth and bounded genus graphs, and specifically, planar graphs.
More generally, graphs of bounded (not necessarily constant) arboricity are of importance, as they
arise in real-world networks and models, such as the world wide web graph, social networks and various random distribution models.

We now summarize our main result.

\begin{theorem}
\label{intro:thm:dyn}
There is a deterministic dynamic algorithm for maintaining a $(\Delta + (4+\epsilon)\alpha)$-edge coloring of an input dynamic graph with maximum degree $\Delta$ and arboricity $\alpha$, with $O(\log^6n /\epsilon^6)$ amortized update time and $O(\log^4 n / \epsilon^5)$ amortized recourse.\footnote{A dynamic algorithm has an {\em amortized update time} (respectively, {\em amortized recourse}) of $O(\lambda)$, if, starting with an empty graph, the total runtime (resp., number of output changes) to handle any sequence of $T$ updates
is $O(T \cdot \lambda)$.}
\end{theorem}

\noindent
Thus, Theorem~\ref{intro:thm:dyn} addresses the above question in the affirmative, for all dynamic graphs with arboricity at most $O(\Delta^{1 - \mu})$, for any constant $\mu > 0$.

An important feature of our dynamic algorithm is that it is {\em adaptive} to changes in the values of $\Delta$ and $\alpha$ over time: At each time-step $t$, we (explicitly)  maintain a proper edge coloring of the input graph $G$ using the colors $\{1, \ldots, \Delta_t + (4+\epsilon)\alpha_t\}$, where $\Delta_t$ and $\alpha_t$ are respectively the maximum degree and arboricity of $G$ at time $t$.

Before giving our full dynamic algorithm, we give a simpler ``warmup'' dynamic algorithm, where we assume access to values $\alpha$ and $\Delta$ such that $\alpha_t \leq \alpha$ and $\Delta_t \leq \Delta$ at each time-step $t$. In this setting, we can maintain a $(\Delta + (4+\epsilon)\alpha)$-edge coloring with $O(\log^2n \log \Delta /\epsilon^2)$ amortized update time and $O(\log n /\epsilon)$ worst-case recourse. As an immediate corollary of our ``warmup'' dynamic algorithm, we also get the following structural result, which should be contrasted with the lower bound of~\cite{ChangHLPU18} for extending partial colorings, which shows that there exist $n$-node graphs of maximum degree $\Delta$ and $(\Delta + c)$-edge colorings on those graphs (for any $c \in [1, \Delta/3]$), such that extending these colorings to color some uncolored edge requires changing the colors of $\Omega(\Delta \log(cn/\Delta)/c)$ many edges.

\begin{corollary}\label{cor:structural}
    Let $G = (V,E)$ be a graph with maximum degree $\Delta$ and arboricity $\alpha$, and let $\chi$ be a $(\Delta + (2 + \epsilon)\alpha)$-edge coloring of $G$. Then, given any uncolored edge $e \in E$, we can extend the coloring $\chi$ so that $e$ is now colored by only changing the colors of $O(\log n / \epsilon)$ many edges.
\end{corollary}

\medskip
\noindent {\bf Independent Work.} In independent and concurrent work, Christiansen, Rotenberg and Vlieghe also obtain a deterministic dynamic algorithm that maintains a $(\Delta + O(\alpha))$-edge coloring in $\tilde O(1)$ amortized update time \cite{eva2023dynamic}.

\subsection{Our Techniques}\label{sec:our tech}

At a high level, our algorithm can be interpreted as a dynamization of a simple static algorithm that computes a $(\Delta + O(\alpha))$-edge coloring of a graph $G$, which can be implemented to run in near-linear time in the static sequential model of computation.\footnote{Recently, \cite{BCPS23} and \cite{kowalik2024edge} considered edge coloring on low arboricity graphs in the static setting, but for the problems of $\Delta + 1$ and $\Delta$ coloring respectively.}
This algorithm is similar to the classic greedy algorithm for $(2\Delta-1)$-edge coloring, which simply scans through all edges of the graph in an arbitrary order and, while scanning any edge $e$, assigns $e$ an arbitrary color in $[2\Delta - 1]$ that has not been already assigned to one of its adjacent edges. Since $e$ has at most $2\Delta - 2$ adjacent edges, such a color must always exist. This static algorithm does something quite similar---the difference is that it computes a `good' ordering of the edges in $G$ instead of using an arbitrary ordering, which allows it to use fewer colors. More specifically, it repeatedly identifies a vertex of minimum degree in $G$, colors an edge incident on in, and removes that edge from the graph.
For the sake of completeness, we include this algorithm and its analysis in Appendix~\ref{sec:static}.
We remark that a variant of this algorithm appears in \cite{podc/BarenboimEM17}, which considers the distributed model of computation.

To highlight the main conceptual insight underlying our approach, we describe the simpler case where $\Delta$ and $\alpha$ are {\em fixed} values (known to the algorithm in advance) that respectively give upper bounds on the maximum degree and arboricity of the input graph at all times. 
We sketch below how to maintain a $(\Delta+ O(\alpha))$-coloring in $\tilde{O}(1)$ update time in this setting. Note that this directly implies a near-linear time static algorithm for $(\Delta+O(\alpha))$-coloring.\footnote{Indeed, we can compute $\Delta$ and a good approximation of $\alpha$ in linear time, and then
simply insert the edges in the input graph into the dynamic algorithm one after another.}
We later outline (Section \ref{sec:adaptive}) how we extend our dynamic algorithm to handle the scenario where $\Delta$ and $\alpha$ change over time.

Our starting point is a well-known ``peeling process'', which leads to a standard decomposition of an input graph $G = (V, E)$ with arboricity at most $\alpha$ \cite{ChibaN85}. The key observation is that {\em any} induced subgraph of $G$ has average degree at most $2\alpha$.\footnote{Indeed, for any subset  $S \subseteq V$, the average degree of $G[S]$ is given by: $2 \cdot |E(G[S])|/|S| \leq 2\alpha$.} Fix any constant $\gamma > 1$.  This motivates the following procedure, which runs for $L = \Theta_{\gamma}(\log n)$ rounds. 

\begin{wrapper}
Initially, during round $1$, we set $Z_1 := V$. Subsequently, during each round $i \in \{2,\dots,L\}$, we find the set of nodes $S \subseteq Z_{i-1}$ that have degree $> 2\gamma \alpha$ in $G[Z_{i-1}]$, and set $Z_i := S$.
\end{wrapper}

\noindent
Consider any given round $i \in [L]$ during the above procedure. Since the subgraph $G[Z_{i-1}]$ has average degree at most $2\alpha$, it follows that at most a $1/\gamma$ fraction of the nodes in $G[Z_{i-1}]$ have degree more than $2\gamma \alpha$. In other words, we get $|Z_{i+1}| \leq |Z_i|/\gamma$, and hence after $L$ iterations we would have $Z_L = \varnothing$. Bhattacharya et al.~\cite{BhattacharyaHNT15} showed how to maintain this decomposition dynamically with $\tilde{O}(1)$ amortized update time, provided that $\gamma > 2$.

Now, our dynamic $(\Delta+O(\alpha))$-coloring algorithm works as follows. Suppose that we are currently maintaining a valid coloring, along with the above decomposition. Upon receiving an update (edge insertion/deletion), we first run the dynamic algorithm of~\cite{BhattacharyaHNT15}, which adjusts the decomposition $Z_1 \supseteq \cdots \supseteq Z_L$, in amortized $\tilde{O}(1)$ time. If the update consisted of an edge deletion, then we do not need to do anything else beyond this point, since the existing coloring continues to remain valid. We next consider the more interesting case, where the update consisted of the insertion of an edge (say) $(u, v)$. 

Let $i \in [L]$ be the largest index such that $(u, v) \in E(G[Z_i])$. Then there must exist some endpoint $x \in \{u, v\}$ that belongs to $Z_i \setminus Z_{i+1}$. W.l.o.g., let $u$ be that endpoint. Since $u \in Z_i \setminus Z_{i+1}$, it follows that the node $u$ has degree at most $2\gamma \alpha$ in $G[Z_i]$. Also, the node $v$ trivially has degree at most $\Delta$ in $G$. Let $E_{(u, v)} \subseteq E$ denote the set of edges $e' \in E$ that belong to one of the following two categories: (I) $e'$ is incident on $u$ and lies in $G[Z_i]$, (II) $e'$ is incident on $v$. We conclude that $|E_{(u, v)}| \leq \Delta+2\gamma \alpha$. Thus, if we have a palette of at least $\Delta+2\gamma\alpha+1 = \Delta + \Theta(\alpha)$ colors, then there must exist a free color in that palette which is not assigned to any edge in $E_{(u, v)}$. Let $c$ be that free color. Using standard binary search data structures, such a color $c$ can be identified in $\tilde{O}(1)$ time~\cite{BhattacharyaCHN18}. We assign the color $c$ to the edge $(u, v)$. This can potentially create a conflict with some other adjacent edge $e'' \in E$ (which might already have been assigned the color $c$). 

However, it is easy to see that such an edge $e''$ must be incident on $u$, i.e., $e'' = (u, y)$ for some $y \in V$, and there must exist some index $i_y < i$ such that $y \in Z_{i_y} \setminus Z_{i_{y}+1}$. We then uncolor the edge $e''$, set $i \leftarrow i_y$, and recolor $e''$ recursively using the same procedure described above. Since after each recursive call, the value of the index $i$  decreases by at least one, this can go on at most $L$ times. This leads to an overall update time of $L \cdot \tilde{O}(1) = \tilde{O}(1)$. See Section~\ref{sec:warmup alg} for details.

\subsubsection{Handling the scenario where $\Delta$ and $\alpha$ change over time} \label{sec:adaptive} We now outline  how we deal with changing values of $\Delta$ and $\alpha$.  Let $\alpha_t$ and $\Delta_t$ respectively denote the arboricity and maximum degree of the input graph $G$ at the current time-step $t$. We need to overcome two technical challenges.

\medskip
\noindent (i) The ``warmup'' algorithm described above works correctly only if it uses a parameter $\alpha \simeq \alpha_t$ to construct the decomposition of $G$. Informally, if $\alpha$ is too small w.r.t.~$\alpha_t$, then the number of iterations $L$ required to construct the decomposition will become huge (possibly infinite, if we aim at achieving
$Z_L = \varnothing$), and this in turn would blow up the update time of the algorithm. In contrast, if $\alpha$ is too large compared to $\alpha_t$, then the algorithm would be using too many colors in its palette.

\medskip
\noindent (ii) After the deletion of an edge $e$, the arboricity $\alpha$ and the maximum degree $\Delta$ of $G$ might decrease. If either parameter drops by a significant amount (across some batch of updates), then we might have to recolor a significant number of edges to ensure that we are still only using $\Delta + O(\alpha)$ many colors, potentially leading to a prohibitively large update time.

\medskip
\noindent
To deal with challenge (i), we generalize the notion of graph decomposition to that of a \emph{decomposition system}. At a high level, a decomposition system is just a collection of graph decompositions, where the relevant parameter across the decompositions is discretized into powers of $(1 + \epsilon)$. This ensures that no matter what the value of $\alpha$ is at the present moment, there is always some decomposition in our system that we can use to extend the coloring. Finally,  to deal with challenge (ii), we ensure that the color of each edge satisfies certain {\em local constraints}, similar to the constraints used to give efficient dynamic algorithms in \cite{BhattacharyaCHN18, Christiansen22}. After the deletion of an edge, we can just uncolor the edges that violate those local constraints, and then recolor them using the decomposition system. However, since the constraints on an edge $e$ depend not just on the degrees of its endpoints but also on the decomposition system, we have to take extra care to ensure that these decompositions don't change too much between updates. See Section~\ref{sec:full algo} for details.

\subsection{Roadmap}
The rest of the paper is organized as follows. Section~\ref{sec:prelim} introduces the relevant preliminary concepts and notations. 
This is followed by Section~\ref{sec:warmup alg}, which contains our warmup dynamic algorithm for fixed $\alpha$. In Section~\ref{sec:full algo}, we present our dynamic algorithm in its full generality. Appendix~\ref{sec:data structs} gives full details of the relevant data structures used by our algorithms in the preceding sections.


\section{Preliminaries}
\label{sec:prelim}

In this section, we define the notations used throughout our paper and describe the notion of \textit{graph decompositions}, which are at the core of our algorithms. We then provide a simple extension of these graph decompositions, which we use as a central component in our final dynamic algorithm.

\subsection{The Dynamic Setting}

In the dynamic setting, we have a graph $G = (V,E)$ that undergoes updates via a sequence of intermixed edge insertions and deletions. Our task is to design an algorithm to explicitly maintain an edge coloring $\chi$ of $G$ as the graph is updated. We assume that the graph $G$ is initially empty, i.e. that the graph $G$ is initialized with $E =\varnothing$. The \emph{update time} of such an algorithm is the time it takes to handle an update, and its \emph{recourse} is the number of edges that change colors while handling an update. More precisely, we say that an algorithm has a worst-case update time of $\lambda$ if it takes at most $\lambda$ time to handle an update, and an \emph{amortized} update time of $\lambda$ if it takes at most $T \cdot \lambda$ time to handle any arbitrary sequence of $T$ updates (starting from the empty graph). Similarly, we say that an algorithm has a worst-case recourse of $\lambda$ if it changes the colors of at most $\lambda$ edges while handling an update, and an \emph{amortized} recourse of $\lambda$ if it changes the colors of at most $T \cdot \lambda$ edges while handling any arbitrary sequence of $T$ updates (starting from the empty graph).

\subsection{Notation}
\label{sec:notation}

Let $G = (V,E)$ be an undirected, unweighted $n$-node graph. Given an edge set $S \subseteq E$, we denote by $G[S]$ the graph $(V, S)$, and given a node set $A \subseteq V$, we denote by $G[A]$ the subgraph induced by $A$, namely $(A, \{(u,v) \in E \, | \, u,v \in A\})$. Given a node $u \in V$ and a subgraph $H$ of $G$, we denote by $N_H(u)$ the set of edges in $H$ that are incident on $u$, and by $\deg_H(u)$ the degree of $u$ in $H$. For an edge $(u,v)$, we define $N_H(u,v)$ to be $N_H(u) \cup N_H(v)$. When we are considering the entire graph $G$, we will often omit the subscripts in $N_G(\cdot)$ and $\deg_G(\cdot)$ and just write $N(\cdot)$ and $\deg(\cdot)$.

\subsection{Graph Decompositions}\label{sec:prelim:decomp}

A central ingredient in our dynamic algorithm is the notion of \emph{$(\beta, d, L)$-decomposition}, defined by Bhattacharya et al.\ \cite{BhattacharyaHNT15}.

\begin{definition}\label{def:decomp}
Given a graph $G = (V,E)$, $\beta \geq 1$, $d \geq 0$, and a positive integer $L$, a $(\beta, d, L)$-decomposition of $G$ is a sequence $(Z_1, \dots ,Z_L)$ of node sets, such that $Z_L \subseteq \dots \subseteq Z_1 = V$ and
$$Z_{i+1} \supseteq \{ u \in Z_i \, | \, \deg_{G[Z_i]}(u) > \beta d\} \textrm{ and } Z_{i+1} \cap \{u \in Z_i \, | \, \deg_{G[Z_i]}(u) < d \} = \varnothing $$
hold for all $i \in [L-1]$.
\end{definition}

\noindent Given a $(\beta, d, L)$-decomposition $(Z_1, \dots ,Z_L)$ of $G = (V,E)$, we abbreviate  $G[Z_i]$ as $G_i$ for all $i$, and for all $u \in V$, we abbreviate
$\deg_{G_i}(u)$ as $\deg_i(u)$ and $N_{G_i}(u)$ as $N_i(u)$. We define $V_i := Z_i \setminus Z_{i+1}$ for all $i \in [L-1]$, and $V_L := Z_L$. We say that $V_i$ is the \emph{$i^{th}$ level} of the decomposition, and define the \emph{level} $\ell(u)$ of any node $u \in V_i$ as $\ell (u) := i$. We define $\deg^+(u) := \deg_{\ell(u)}(u)$ and $N^+(u) := N_{\ell(u)}(u)$ for $u \in V$.
Given an edge $e=(u,v)$, we define the level $\ell(e)$ of $e$ as $\ell(e) := \min\{\ell(u), \ell(v)\}$. Note also that for all $u \in V \setminus V_{L}$, $\deg^+(u) \leq \beta d$. However, given some $u \in V_L$, $\deg^+(u)$ may be much larger than $\beta d$, which motivates the following useful fact concerning such decompositions.

\begin{lemma}[\cite{BhattacharyaHNT15}]
\label{lem: Z empty}
    Let $G = (V,E)$ be an arbitrary graph with arboricity $\alpha$, let $\beta$, $\epsilon$, $d$ be any parameters such that $\beta \geq 1$, $0 < \epsilon < 1$, $d \geq 2(1 + \epsilon)\alpha$, and let $L = 2 + \lceil \log_{(1 + \epsilon)}n \rceil$. Then for any $(\beta, d, L)$-decomposition $(Z_1,...,Z_L)$ of $G$, it holds that $Z_L = \varnothing$.
\end{lemma}

\begin{proof}
    Let $(Z_1,...,Z_L)$ be a $(\beta, d, L)$-decomposition of $G$ satisfying the conditions of the lemma. Let $i$ be an arbitrary index in $[L-1]$. Since the arboricity of $G_i$ is at most $\alpha$, the average degree in $G_i$ is at most $2 \alpha$. On the other hand, by definition, the degree of any node in $Z_{i+1}$ in the graph $G_i$ is at least $d \geq 2(1 + \epsilon)\alpha$. It follows that
    $$2(1+\epsilon)\alpha |Z_{i+1}| \leq \sum_{u \in Z_{i+1}} \deg_{i}(u) \leq \sum_{u \in Z_{i}} \deg_{i}(u) \leq 2\alpha |Z_{i}|, $$
    and hence $|Z_{i+1}| \leq |Z_i|/(1 + \epsilon)$. Inductively, we obtain $|Z_L| \leq (1 + \epsilon)^{1-L}|Z_1| \leq 1/(1 + \epsilon) < 1$, yielding $Z_L = \varnothing$.
\end{proof}

\noindent
\textbf{Orienting the Edges.}
For our purposes, it will be useful to think of a decomposition of $G$ as inducing an \emph{orientation} of the edges. In particular, given an edge $e = (u,v)$, \emph{we orient the edge from the endpoint of lower level towards the endpoint of higher level}. If the two endpoints have the same level, we orient the edge arbitrarily. We write $u \prec v$ to denote that the edge $e$ is oriented from $u$ to $v$. Note that $\deg^+(u)$ is an upper bound on the out-degree of $u$ with respect to this orientation of the edges.

\medskip
\noindent \textbf{Dynamic Decompositions.} Bhattacharya et al.\ give a deterministic fully dynamic data structure that can be used to explicitly maintain a $(\beta, d, L)$-decomposition of a graph $G = (V,E)$ under edge updates with small amortized update time. This algorithm also has small amortized recourse, where the recourse of an update is defined as the number of edges that change level following the update. The following theorem, from  Section 4.1 of \cite{BhattacharyaHNT15}, will be used as a black box in our dynamic algorithm.

\begin{proposition}[\cite{BhattacharyaHNT15}]
\label{prop:dynamic-decomp}
For any constant $\beta \geq 2 + 3 \epsilon$, there is a deterministic fully-dynamic algorithm that maintains a $(\beta, d, L)$-decomposition of a graph $G = (V,E)$ with amortized update time and amortized recourse both bounded by $O(L/\epsilon)$.
\end{proposition}

\noindent
It is straightforward to modify this dynamic algorithm to explicitly maintain the orientation of the edges that we described above without changing its asymptotic behavior. Furthermore, we can assume that the orientation of an edge changes \emph{only} when it changes level.

\subsection{Graph Decomposition Systems}

In order for our dynamic algorithm to be able to deal with dynamically changing arboricity $\alpha$, we will need to give a slight generalization of Definition~\ref{def:decomp}, which we refer to as a \textit{decomposition system}. Intuitively, this will enable us to maintain multiple decompositions, one for each `guess' of the arboricity, allowing us to use whichever decomposition is most appropriate to modify the edge coloring while handling an update.

\begin{definition}
     Given a graph $G = (V,E)$, $\beta \geq 1$, a sequence $(d_j)_{j \in [K]}$ such that $d_j \geq 0$, and a positive integer $L$, a $(\beta, (d_j)_{j \in [K]}, L)$-decomposition system of $G$ is a sequence $(Z_{i,j})_{i \in [L], j \in [k]}$ of node sets, where for each $j \in [K]$, $(Z_{i,j})_{i \in [L]}$ is a $(\beta, d_j, L)$-decomposition of $G$.
\end{definition}

\noindent
Given a $(\beta, (d_j)_{j \in [K]}, L)$-decomposition system of $G=(V,E)$, we denote the graph $G[Z_{i,j}]$ by $G_{i,j}$, $\deg_{G_{i,j}}(u)$ by $\deg_{i,j}(u)$, and $N_{G_{i,j}}(u)$ by $N_{i,j}(u)$ for $u \in V$.
We say that $(Z_{i,j})_i$ is the \emph{$j^{th}$ layer} of the decomposition system. We denote by $\ell_j(u)$ the level of node $u$ in the decomposition $(Z_{i,j})_i$ and define $\deg_j^+(u) := \deg_{\ell_j(u),j}(u)$ and $N_j^+(u) := N_{\ell_j(u),j}(u)$ for $u \in V$.

Given a node $u$, we define the \textit{layer} of $u$ as $\mathcal L (u) = \min\{j \in [K] \, | \, \ell_j(u) < L\}$.
Given an edge $e = (u,v)$, we define the layer of $e$ as $\mathcal L(e) = \min\{\mathcal{L}(u), \mathcal L(v)\}$. We denote the orientation of the edges induced by the decomposition $(Z_{i,j})_i$ by $\prec_j$.

We can use the data structure from Proposition~\ref{prop:dynamic-decomp} to dynamically maintain a decomposition system, giving us the following proposition. In this context, we define the recourse of an update to be the number of edges that change levels in \emph{some} layer.

\begin{proposition}\label{prop:dynamic-decomp-sys}
For any constant $\beta \geq 2 + 3 \epsilon$, there is a deterministic fully dynamic algorithm that maintains a $(\beta, (d_j)_{j \in [K]}, L)$-decomposition system of a graph $G = (V,E)$ with amortized update time and amortized recourse $O(KL/\epsilon)$.
\end{proposition}

\noindent As before, we assume that the orientation of an edge $e$ with respect to $\prec_j$ changes only when $\ell_{j}(e)$ changes.


\section{A Warmup Dynamic Algorithm (for Fixed $\alpha$)}\label{sec:warmup alg}

We now turn our attention towards designing an algorithm that can \emph{dynamically} maintain a $(\Delta + O(\alpha))$-edge coloring of the graph $G$ as it changes over time. A starting point for creating such an algorithm is the static algorithm from Appendix~\ref{sec:static} that we outline in Section~\ref{sec:our tech}. Unfortunately, the highly sequential nature of this algorithm makes it very challenging to dynamize directly, as it is not clear how to efficiently maintain the output in the dynamic setting. In order to overcome this obstacle, we use the notion of graph decompositions (see Section~\ref{sec:prelim:decomp}). Informally, these graph decompositions can be interpreted as an `approximate' version of the sequence in which the static algorithm colors the edges in the graph---where instead of peeling off a node with smallest degree one at a time, we peel off large batches of nodes with sufficiently small degrees simultaneously. This leads to a `more robust' structure that can be maintained dynamically in an efficient manner.

Let $G = (V,E)$ be a dynamic graph that undergoes updates via edge insertions and deletions. In this section, we work in a simpler setting where we assume that we are given an $\alpha$ and are guaranteed that the maximum arboricity of the graph $G$ remains at most $\alpha$ throughout the entire sequence of updates. We then give a deterministic fully dynamic algorithm that maintains a $(\Delta + O(\alpha))$-edge coloring of $G$, where $\Delta$ is an upper bound on the maximum degree of $G$ at any point throughout the entire sequence of updates.\footnote{Note that the algorithm needs prior knowledge of $\alpha$, but not $\Delta$.} Without dealing with implementation details, we show that it achieves $\tilde O(1)$ worst-case recourse per update. In Section~\ref{sec:full algo}, we extend our result to the setting where $\Delta$ and $\alpha$ are not bounded and show how to maintain a $(\Delta + O(\alpha))$-edge coloring of $G$ where $\alpha$ and $\Delta$ are the \textit{current} arboricity and maximum degree of $G$ respectively and change over time.

\subsection{Algorithm Description}

For the rest of this section, fix some constants $\epsilon$, $\beta$, and $L$ such that: $0 < \epsilon < 1$, $\beta = 2 + 3\epsilon$, $L = 2 + \lceil \log_{1 + \epsilon}n \rceil$. At a high level, our algorithm works by dynamically maintaining a $(\beta, 2(1 + \epsilon)\alpha, L)$-decomposition $(Z_i)_{i=1}^L$ of the graph $G$ by using Proposition~\ref{prop:dynamic-decomp}. During an update, our algorithm first updates the decomposition $(Z_i)_i$, and then uses this decomposition to find a path of length at most $L$ such that, by only changing the colors assigned to the edges in this path, it can update the coloring to be valid for the updated graph. Since $L = \tilde O(1)$, this immediately implies the worst-case recourse bound. Algorithm \ref{alg:init} gives the procedure that we call to initialize our data structure, creating a decomposition of the empty graph, and Algorithms \ref{alg:insertion} and \ref{alg:deletion} give the procedures called when handling insertions and deletions respectively.
\begin{algorithm}
    \SetAlgoLined
    \DontPrintSemicolon
    \KwIn{An empty graph $G=(V, \varnothing)$ and a parameter $\alpha$}
    Create a $(\beta, 2(1 + \epsilon)\alpha, L)$-decomposition $(Z_i)_{i \in [L]}$ of $G$\;
    \caption{\textsc{Initialize}$(G, \alpha)$}
    \label{alg:init}
\end{algorithm}
\begin{algorithm}
    \SetAlgoLined
    \DontPrintSemicolon
    \KwIn{An edge $e$ to be inserted into $G$}
    Insert the edge $e$ into $G$\;
    $\chi(e) \leftarrow \perp$\;
    Update the $(\beta, 2(1 + \epsilon)\alpha, L)$-decomposition $(Z_i)_{i}$ of $G$\;
    \textsc{ExtendColoring}$(e, (Z_i)_{i})$\;
    \caption{\textsc{Insert}$(e)$}
    \label{alg:insertion}
\end{algorithm}
\begin{algorithm}
    \SetAlgoLined
    \DontPrintSemicolon
    \KwIn{An edge $e$ to be deleted from $G$}
    Delete the edge $e$ from $G$\;
    $\chi(e) \leftarrow \perp$\;
    Update the $(\beta, 2(1 + \epsilon)\alpha, L)$-decomposition $(Z_i)_{i}$ of $G$\;
    \caption{\textsc{Delete}$(e)$}
    \label{alg:deletion}
\end{algorithm}
\begin{algorithm}
    \SetAlgoLined
    \DontPrintSemicolon
    \KwIn{An uncolored edge $e$ and a $(\beta, 2(1 + \epsilon)\alpha, L)$-decomposition $(Z_i)_{i}$ of $G$}
    $S \leftarrow \{e\}$\;
    \While{$S \neq \varnothing$} {
        Let $f = (u,v)$ be any edge in $S$ where $u \prec v$\;
        $C^+_u \leftarrow \chi (N^+(u))$\;
        $C_v \leftarrow \chi (N(v))$\;
        Set $c$ to any element in $[|C^+_u| + |C_v| + 1] \setminus (C^+_u \cup C_v)$\;
        \If{$c \in \chi (N(u))$} {
            Let $f'$ be the edge in $N(u)$ with $\chi(f') = c$\;
            $\chi(f') \leftarrow \perp$ and $S \leftarrow S \cup \{ f' \}$\;
        }
        $\chi(f) \leftarrow c$ and $S \leftarrow S \setminus \{ f \}$\;
    }
    \caption{\textsc{ExtendColoring}$(e, (Z_i)_{i})$}
    \label{alg:extend}
\end{algorithm}

\noindent The following theorem, which we prove next, summarizes the behavior of our warmup dynamic algorithm.

\begin{theorem}\label{thm:warmup}
    The warmup dynamic algorithm is deterministic and, given a sequence of updates for a dynamic graph $G$ and a value $\alpha$ such that the arboricity of $G$ never exceeds $\alpha$, maintains a $(\Delta + (4 + \epsilon)\alpha)$-edge coloring, where $\Delta$ is the maximum degree of $G$ throughout the entire sequence of updates. The algorithm has $O(\log n/ \epsilon)$ worst-case recourse per update and $O(\log^2 n \log \Delta / \epsilon^2)$ amortized update time.
\end{theorem}

\subsection{Analysis of the Warmup Algorithm}

We now show that the warmup algorithm maintains a $(\Delta + 2\beta(1 + \epsilon)\alpha)$-edge coloring and has a worst-case recourse of at most $L = O(\log n / \epsilon)$ per update.\footnote{Note that $2\beta(1 + \epsilon)\alpha = (4 + O(\epsilon))\alpha$.}

\begin{lemma}\label{lem:dynamic anal 1}
    Let $G=(V,E)$ be a graph with maximum degree at most $\Delta$ and arboricity at most $\alpha$. Let $e$ be an edge in $G$, $(Z_i)_i$ a $(\beta, 2(1 + \epsilon)\alpha, L)$-decomposition of $G$ and $\chi$ a $(\Delta + 2\beta(1 + \epsilon)\alpha)$-edge coloring of $G - e$. Then running \textsc{ExtendColoring}$(e, (Z_i)_{i})$:
    \begin{enumerate}
        \item changes the colors of at most $L$ edges in $G$, and
        \item turns $\chi$ into a $(\Delta + 2\beta(1 + \epsilon)\alpha)$-edge coloring of $G$.
    \end{enumerate}
\end{lemma}

\begin{proof}
    We first prove (1). Let $e_i$ denote the edge that is uncolored at the start of the $i^{th}$ iteration of the while loop as we run the procedure. Let $\ell (e_i)$ denote the minimum of the level of both of its endpoints. Clearly $\ell(e_1) \leq L$ since this is the highest level and $\ell(e_i) \geq 1$ for all $i$ since this is the lowest level. Suppose the while loop iterates at least $i$ times for some integer $i \geq 2$. Let $e_{i-1} = (u,v)$ where $u \prec v$, and hence $\ell(u) \leq \ell(v)$ (see Section~\ref{sec:prelim:decomp}). Since $e_{i} \in N(u)$ during iteration $i-1$ but $\chi(e_{i}) \notin \chi (N^+(u))$, we have that $e_{i} \notin N^+(u)$, and hence the endpoint of $e_{i}$ that is not $u$ appears in a level strictly below the level of $u$, so $\ell(e_{i}) < \ell(e_{i-1})$. It follows that $1 \leq \ell(e_i) \leq L + 1 - i$, so the while loop iterates at most $L$ times. For (2), note that if we let $e_i = (u,v)$ where $u \prec v$, then $|C^+_u| = \deg^+(u) - 1$ and $|C_v| = \deg(v) - 1$, so
    $$|C^+_u| + |C_v| + 1 \leq \deg^+(u) + \deg(v) - 1 \leq \Delta + 2\beta(1 + \epsilon)\alpha, $$
    and so the procedure never assigns any $e_i$ a color larger than $\Delta + 2\beta(1 + \epsilon)\alpha$. Since we know from (1) that the procedure terminates after at most $L$ iterations, after which every edge in the graph is colored, and $\chi$ was a $(\Delta + 2\beta(1 + \epsilon)\alpha)$-edge coloring of the graph $G - e_1$ at the start of the procedure, it follows by induction that after the procedure terminates $\chi$ assigns each edge in $G$ a color from $[\Delta + 2\beta(1 + \epsilon)\alpha]$. Furthermore, our algorithm can only terminate if this assignment forms a valid edge coloring. Hence, $\chi$ is a $(\Delta + 2\beta(1 + \epsilon)\alpha)$-edge coloring of $G$.
\end{proof}

\begin{lemma}
    The warmup algorithm maintains a $(\Delta + 2\beta(1 + \epsilon)\alpha)$-edge coloring of the graph.
\end{lemma}

\begin{proof}
    We prove this by induction. Since $G$ is initially empty, the empty map is trivially a coloring of $G$. Let $\lambda = \Delta + 2\beta(1 + \epsilon)\alpha$. Suppose $\chi$ is a $\lambda$-edge coloring of $G$ after the $i^{th}$ update. If the $i+1^{th}$ update is a deletion, $\chi$ is still a $\lambda$-edge coloring of the updated graph and we are done. If the $i+1^{th}$ update is an insertion, then we run Algorithm~\ref{alg:extend} in order to update $\chi$. By part (2) of Lemma~\ref{lem:dynamic anal 1}, it follows that $\chi$ is a $\lambda$-edge coloring of the updated graph once the procedure terminates.
\end{proof}

\begin{lemma}
    The warmup algorithm changes the colors of at most $L$ edges while handling an update.
\end{lemma}

\begin{proof}
    While handling the deletion of an edge $e$, our algorithm uncolors the edge $e$ and does not change the color of any other edge. While handling the insertion of an edge $e$, our algorithm only changes the colors of edges while handling the call to \textsc{ExtendColoring}$(e, (Z_i)_{i})$. By part (1) of Lemma~\ref{lem:dynamic anal 1}, this changes the colors of at most $L$ edges.
\end{proof}

\noindent In Appendix~\ref{app:warmup dynamic}, we give the proof of the following lemma using the data structures presented in Appendix~\ref{sec:coloring data structs}.

\begin{lemma}\label{lem:warmup dynamic update time}
    The warmup algorithm has an amortized update time of $O(\log^2 n \log \Delta / \epsilon^2)$.
\end{lemma}


\noindent
We also note that Corollary~\ref{cor:structural} follows immediately from Lemma~\ref{lem:dynamic anal 1}. In particular, if we set $\beta = 1$, by Lemma~\ref{lem: Z empty}, the proof Lemma~\ref{lem:dynamic anal 1} still holds. Hence, we can use \textsc{ExtendColoring} along with \emph{any} $(1, 2(1 + \epsilon)\alpha, L)$-decomposition of $G$ in order to extend any $(\Delta + 2(1 + \epsilon)\alpha)$-edge coloring $\chi$ with an uncolored edge $e$ so that the edge $e$ is now colored by only changing the colors of $O(\log n / \epsilon)$ many edges.


\section{The Dynamic Algorithm}\label{sec:full algo}

We now describe our full dynamic algorithm and show that it maintains a $(\Delta + O(\alpha))$-edge coloring of the graph. We then use Proposition~\ref{prop:dynamic-decomp-sys} to show that we can get $\tilde O(1)$ amortized recourse. Finally, in Appendix~\ref{sec:data structs}, we describe the relevant data structures and use them to implement our algorithm to get $\tilde O(1)$ amortized update time.

\subsection{Algorithm Description}

In order to describe our algorithm, we fix some constant $\epsilon$ such that $0 < \epsilon < 1$
and set 
$\beta = 2 + 3\epsilon$, $L = 2 + \lceil \log_{1 + \epsilon}n \rceil$. Let $\tilde \alpha_j := (1 + \epsilon)^{j-1}$ and note that, for any $n$-node graph $G$ with arboricity $\alpha$, $\tilde \alpha_1 = 1 \leq \alpha \leq n < \tilde \alpha_L$.

\medskip
\noindent \textbf{Informal Description.} Our algorithm works by maintaining the invariant that each edge $e = (u,v)$ receives a color in the set $[\deg(v) + O(\tilde \alpha_{\mathcal L(e)})]$, where $u \prec_{\mathcal L(e)} v$. 
Since $\deg(v) \leq \Delta$ and $\tilde \alpha_{\mathcal L (e)} = O(\alpha)$ (see Lemma~\ref{lem: alphaL small}), it follows that the algorithm uses at most $\Delta + O(\alpha)$ many colors. When an edge is inserted or deleted, this may cause some $\tilde O(1)$ many edges to violate the invariant. We begin by first identifying all such edges and uncoloring them. We then update the decomposition system maintained by our algorithm, which may again cause some $\tilde O(1)$ many edges (on average) to violate the invariant. We again identify and uncolor all such edges. We now want to color each of the uncolored edges, while ensuring that we satisfy this invariant at all times. We do this by using the decomposition system maintained by our algorithm: we take an uncolored edge $f = (u,v)$ such that $u \prec_{\mathcal L(f)} v$ 
and assign it a color $c$ that is \textit{not} assigned to any of the edges in $N_{\mathcal L(f)}^+(u)$ or $N(v)$. If there is an edge $f'$ adjacent to $f$ that is also colored with $c$, we uncolor this edge. We repeat this process iteratively until all edges are colored. We can show that: (1) there are at most $\deg(v) + O(\tilde \alpha_{\mathcal L(f)})$ many edges in $N_{\mathcal L(f)}^+(u) \cup N(v)$, and hence we can find such a $c$ in the palette $[\deg(v) + O(\tilde \alpha_{\mathcal L(f)})]$, and (2) if there is such an edge $f'$ adjacent to $f$ that is also colored with $c$, then either $\ell_{\mathcal L(f')}(f') < \ell_{\mathcal L(f)}(f)$ or $\mathcal L(f') < \mathcal L(f)$, allowing us to carry out a potential function argument that shows that the process terminates with all edges colored after $\tilde O(1)$ iterations on average, giving us an amortized recourse bound.

\medskip
\noindent \textbf{Formal Description.} The following pseudo-code gives a precise formulation of our algorithm.

\begin{algorithm}
    \SetAlgoLined
    \DontPrintSemicolon
    \KwIn{An empty graph $G=(V, \varnothing)$}
    Create a $(\beta, (2(1 + \epsilon)\tilde \alpha_j)_{j \in [L]}, L)$-decomposition system $(Z_{i,j})_{i,j \in [L]}$ of $G$\;
    \caption{\textsc{Initialize}$(G)$}
    \label{alg:init 2}
\end{algorithm}
\begin{algorithm}
    \SetAlgoLined
    \DontPrintSemicolon
    \KwIn{An edge $e$ to be inserted into $G$}
    Insert the edge $e$ into $G$\;
    $S \leftarrow \textsc{UpdateDecompositions}(e)$\;
    $\chi(f) \leftarrow \perp$ for all $f \in S$\; 
    \textsc{ExtendColoring}$(S)$\;
    \caption{\textsc{Insert}$(e)$}
    \label{alg:insertion 2}
\end{algorithm}
\begin{algorithm}
    \SetAlgoLined
    \DontPrintSemicolon
    \KwIn{An edge $e$ to be deleted from $G$}
    Delete the edge $e$ from $G$\;
    $S \leftarrow \varnothing$\;
    \For{$v \in e$}{
        $S \leftarrow S \cup \{f = (u,v) \in N(v) \, | \, u \prec_{\mathcal L(f)} v \textnormal{ and } \chi(f) > \deg(v) + 2\beta (1 + \epsilon) \tilde \alpha_{\mathcal L(f)} \}$\;
    }
    $S \leftarrow S \cup \textsc{UpdateDecompositions}(e)$\;
    $\chi(f) \leftarrow \perp$ for all $f \in S$\; 
    \textsc{ExtendColoring}$(S)$\;
    \caption{\textsc{Delete}$(e)$}
    \label{alg:deletion 2}
\end{algorithm}
\begin{algorithm}
    \SetAlgoLined
    \DontPrintSemicolon
    \KwIn{The edge $e$ that has been inserted/deleted from $G$}
    Update the decomposition system $(Z_{i,j})_{i,j}$\;
    Let $S' \subseteq E$ be the set of all edges whose level changes in some layer\;
    \Return{$S'$}
    \caption{\textsc{UpdateDecompositions}$(e)$}
    \label{alg:update 2}
\end{algorithm}
\begin{algorithm}
    \SetAlgoLined
    \DontPrintSemicolon
    \KwIn{A set $S$ of uncolored edges}
    \While{$S \neq \varnothing$} {
        Let $f = (u,v)$ be any edge in $S$ where $u \prec_{\mathcal L(f)} v$\;
        $C^+_u \leftarrow \chi (N_{\mathcal L(f)}^+(u))$\;
        $C_v \leftarrow \chi (N(v))$\;
        Let $c$ be any element in $[|C^+_u| + |C_v| + 1] \setminus (C^+_u \cup C_v)$\;
        \If{$c \in \chi (N(u))$} {
            Let $f'$ be the edge in $N(u)$ with $\chi(f') = c$\;
            $\chi(f') \leftarrow \perp$ and $S \leftarrow S \cup \{ f' \}$\;
        }
        $\chi(f) \leftarrow c$ and $S \leftarrow S \setminus \{ f \}$\;
    }
    \caption{\textsc{ExtendColoring}$(S)$}
    \label{alg:extend 2}
\end{algorithm}

\noindent
The following theorem, which we prove next, summarizes the behavior of our full dynamic algorithm.

\begin{theorem}\label{thm:main}
     The dynamic algorithm is deterministic and, given a sequence of updates for a dynamic graph $G$, maintains a $(\Delta + (4 + \epsilon)\alpha)$-edge coloring, where $\Delta$ and $\alpha$ are the dynamically changing maximum degree and arboricity of $G$, respectively. The algorithm has $O(\log^4n/\epsilon^5)$ amortized recourse per update and $O(\log^5n \log \Delta/\epsilon^6)$ amortized update time.\footnote{Whenever the term $\Delta$ appears in an amortized bound, this should be interpreted as being an upper bound on the maximum degree across the whole sequence of updates. In the introduction, we replaced the $\log \Delta$ term with $\log n$ for simplicity.}
\end{theorem}

\noindent We split the proof of Theorem~\ref{thm:main} into two part. In Section~\ref{sec:anal}, we show that our dynamic algorithm maintains a $(\Delta + 2\beta(1 + \epsilon)^2\alpha)$-edge coloring and has an amortized recourse of $O(\log^4 n /\epsilon^5)$.\footnote{Note that $2\beta(1 + \epsilon)^2\alpha = (4 + O(\epsilon))\alpha$.} In Appendix~\ref{sec:data structs}, we describe the data structures used by our algorithm, before showing how to use them in order to get $O(\log^5n \log \Delta/\epsilon^6)$ amortized update time.

\subsection{Analysis of the Dynamic Algorithm}\label{sec:anal}

For the rest of Section~\ref{sec:anal}, fix a dynamic graph $G = (V,E)$,
and a $(\beta, (2(1 + \epsilon)\tilde \alpha_j)_j, L)$-decomposition system $\mathcal Z = (Z_{i,j})_{i,j}$ of $G$.
Recall that $\epsilon$ is a fixed constant with $0 < \epsilon < 1$,
and that 
$\beta = 2 + 3\epsilon$, $L = 2 + \lceil \log_{1 + \epsilon}n \rceil$.

We begin with the following simple observations.

\begin{lemma}\label{lem: L <= K}
    For all nodes $u \in V$, we have that $\mathcal L (u) \leq j^\star$, where $j^\star \in [L]$ is the unique value such that $\alpha \leq \tilde \alpha_{j^\star} < (1 + \epsilon) \alpha$.
\end{lemma}

\begin{proof}
    By Lemma \ref{lem: Z empty}, we know that $Z_{L,j^\star} = \varnothing$. Hence, $\mathcal L(u) \leq j^\star$ for every node $u \in V$.
\end{proof}


\begin{corollary}\label{lem: alphaL small}
    For all edges $e \in E$, we have that $\tilde \alpha_{\mathcal L(e)} < (1 + \epsilon) \alpha$.
\end{corollary}

\noindent We now define the notation of a \emph{good} edge coloring. In such an edge coloring, the colors satisfy certain locality constraints, which makes it easier to maintain dynamically.

\begin{definition}
     Given an edge coloring $\chi$ of the graph $G$, we say that $\chi$ is a \emph{good} edge coloring of $G$ with respect to the decomposition system $\mathcal Z$ if and only if for every edge $e = (u,v) \in E$ such that $\chi(e) \neq \perp$ and $u \prec_{\mathcal L(e)} v$, we have that $\chi(e) \leq \deg(v) + 2\beta(1 + \epsilon)\tilde \alpha_{\mathcal L(e)}$.
\end{definition}

\noindent The following lemma shows that our algorithm can be used to maintain a good edge coloring.

\begin{lemma}\label{lem:full dynamic anal 1}
    Let $\chi$ be a good edge coloring of the graph $G$ w.r.t. $\mathcal Z$ and let $S \subseteq E$ be the set of edges that are left uncolored by $\chi$. Then running \textsc{ExtendColoring}$(S)$:
    \begin{enumerate}
        \item changes the colors of at most $L^2|S|$ edges in $G$, and
        \item turns $\chi$ into a good edge coloring with no uncolored edges.
    \end{enumerate}
\end{lemma}

\begin{proof}
    We begin by proving (1). Given some edge $f$, define the potential of $f$ by 
    $$\Psi(f) = L (\mathcal L(f) - 1) + \ell_{\mathcal L(f)}(f).$$
    Given the set of edges $S$, define the potential of $S$ as $\Psi(S) = \sum_{f \in S} \Psi(f)$. By Lemma~\ref{lem: L <= K}, we have that, for any edge $f$, $1 \leq \Psi(f) = L (\mathcal L(f) - 1) + \ell_{\mathcal L(f)}(f) \leq L (L - 1) + L = L^2$. Hence, $|S| \leq \Psi(S) \leq L^2|S|$. During each iteration of the while loop in Algorithm~\ref{alg:extend 2}, exactly one edge receives a new color (and at most one edge becomes uncolored). We now show that during each iteration of the loop, $\Psi(S)$ drops by at least one, implying that we have at most $L^2|S|$ iterations in total, changing the colors of at most $L^2|S|$ many edges. Let $f$ be the edge in $S$ that we are coloring during some iteration of the loop and let $c$ be the color that it receives. During the iteration, we remove $f$ from $S$; furthermore, if there exists some edge $f'$ colored with $c$ that shares an endpoint with $f$, we uncolor $f'$ and place it in $S$. If there is no such edge $f'$, then $\Psi(S)$ drops by at least $1$ since we remove $f$ from $S$ and $\Psi(f) \geq 1$. Suppose that there is such an edge $f'$. We now argue that $\Psi(f') < \Psi(f)$. We first note that one of the endpoints of $f'$ is not contained in $Z_{i,j}$ where $i = \ell_{\mathcal L(f)}(u)$ and $j = \mathcal L(f)$. This implies that $\ell_{\mathcal L(f)}(f') < \ell_{\mathcal L(f)}(f)$, so $\mathcal L(f') \leq \mathcal L(f)$. Hence, if $\mathcal L(f) = \mathcal L(f')$, it follows that $\Psi(f') < \Psi(f)$. Otherwise, $\mathcal L(f') < \mathcal L(f)$, and we have that
    $$\Psi(f) - \Psi(f') = L(\mathcal L(f) - \mathcal L(f')) + \ell_{\mathcal L(f)}(f) - \ell_{\mathcal L(f')}(f') \geq L + (1 - L) \geq 1.$$
    In either case, $\Psi(S)$ drops by at least $1$. We now prove (2). Let $f = (u,v)$ be the edge in $S$ that we are coloring during some iteration of the while loop such that $u \prec_{\mathcal L(f)} v$.
    We need to show that the color $c$ picked by the algorithm satisfies $c \leq \deg(v) + 2\beta(1 + \epsilon)\tilde \alpha_{\mathcal L(f)}$. It will then follow by induction that the coloring produced by calling \textsc{ExtendColoring}$(S)$ is good given that we start with a good coloring. We first note that $|C_v| \leq \deg(v) - 1$. Now note that $|C^+_u| \leq \deg_{\mathcal L(f)}^+(u) - 1$. Since $\deg_{\mathcal L(f)}^+(u) \leq 2 \beta (1 + \epsilon)\tilde \alpha_{\mathcal L(f)}$, we get the desired bound on $c$. Finally, note that at the start of each iteration, the uncolored edges correspond to exactly the edges in $S$. Since the algorithm terminates if and only if $S = \varnothing$ and we know that the algorithm terminates after at most $L^2|S|$ many iterations, it follows that the resulting coloring has no uncolored edges.
\end{proof}

\begin{lemma}
    The dynamic algorithm maintains a $(\Delta + 2\beta(1 + \epsilon)^2\alpha)$-edge coloring of the graph.
\end{lemma}

\begin{proof}
    By showing that our algorithm maintains a good edge coloring, it follows by Corollary \ref{lem: alphaL small} that, for any edge $e \in E$, we have $\chi(e) \leq \Delta + 2\beta(1 + \epsilon) \tilde \alpha_{\mathcal L(e)} \leq \Delta + 2(1 + \epsilon)^2 \alpha$. We do this by showing that, after an update, the algorithm uncolors all of the edges $f=(u,v)$ in the graph that don't satisfy the condition $\chi(f) \leq \deg(v) + 2\beta(1 + \epsilon)\tilde \alpha_{\mathcal L(f)}$ for $u \prec_{\mathcal L(f)} v$
    in the updated decomposition system, places them in a set $S$, and calls Algorithm \ref{alg:extend 2} on the set $S$. By Lemma~\ref{lem:full dynamic anal 1}, it then follows that the algorithm maintains a good coloring of the entire graph.

    We refer to an edge $e=(u,v)$ as \emph{bad} if it does not satisfy the condition required by a good coloring, i.e. if $\chi(e) \neq \perp$ and $\chi(e) > \deg(v) + 2\beta(1 + \epsilon)\tilde \alpha_{\mathcal L(f)}$ where $u \prec_{\mathcal L(f)} v$. Suppose we have a good edge coloring of the entire graph and insert an edge $e$ into the graph. Since this cannot decrease the degrees of any nodes or change the levels of any edges (since we have not yet updated the decomposition system) this cannot cause any edges to become bad. On the other hand, if we delete an edge $e = (u,v)$, some of the edges incident to $u$ and $v$ might become bad since $\deg(u)$ and $\deg(v)$ decrease by $1$. Any such edges that become bad must be contained within the set $\Gamma_u \cup \Gamma_v$ where
    $$\Gamma_w = \{f = (w',w) \in N(w) \, | \, w' \prec_{\mathcal L(f)} w \textnormal{ and } \chi(f) > \deg(w) + 2\beta (1 + \epsilon) \tilde \alpha_{\mathcal L(f)} \}$$
    where the degrees are w.r.t. the state of the graph $G$ \textit{after} the deletion of $e$. If we uncolor all of the edges in $\Gamma_u \cup \Gamma_v$, we restore $\chi$ to being a good edge coloring. After updating the decomposition system, the levels of some edges might change in some layers. Any edge that does not change levels in any layer will \emph{not} become bad, since $\mathcal L(f)$ (and hence $\tilde \alpha_{\mathcal L(f)}$) and its orientation in $\prec_{\mathcal L(f)}$ do not change.
    However, an edge $f$ that changes levels in some layer might become bad if $\mathcal L(f)$ decreases (causing the value of $\tilde \alpha_{\mathcal L(f)}$ to decrease) or if its orientation with respect to $\prec_{\mathcal L(f)}$ changes.
    Hence, we uncolor all such edges.\footnote{Note that these are precisely the edges that contribute towards the recourse of the dynamic decomposition system.} This guarantees that there are no bad edges when we call \textsc{ExtendColoring}. Since we give \textsc{ExtendColoring} all of the edges that are uncolored, it follows that we maintain a good edge coloring of the entire graph.
\end{proof}

\begin{lemma}\label{lem:full algorithm recourse}
    The dynamic algorithm has $O(\log^4 n / \epsilon^5)$ amortized recourse per update.
\end{lemma}

\begin{proof}
    Suppose that our algorithm handles a sequence of $T$ updates (edge insertions or deletions) starting from an empty graph. Let $S^{(t)}$ denote the set of edges uncolored by our algorithm during the $t^{th}$ update before calling \textsc{ExtendColoring} on the set $S^{(t)}$. By Lemma~\ref{lem:full dynamic anal 1}, we know that at most $L^2|S^{(t)}| = O(|S^{(t)}|\log^2 n/\epsilon^2)$ many edges will change color during this update. By showing that $(1/T) \cdot \sum_{t \in [T]} |S^{(t)}|$ is $O(\log^2 n/ \epsilon^3)$, our claimed amortized recourse bound follows. Now fix some $t \in [T]$ and let $e = (u,v)$ be the edge being either inserted or deleted during this update. The edges uncolored by the algorithm while handling this update are either contained in the set $\Gamma_u \cup \Gamma_v$ (if the update is a deletion) or change levels in some layer after we update the decomposition system. There can only be at most $2L$ many edges of the former type. This is because, given some $j \in [L]$, there is at most one edge $f \in \Gamma_w$ with $\mathcal L(f) = j$ such that $\chi(f) > \deg(w) + 2\beta (1 + \epsilon) \tilde \alpha_{\mathcal L(f)}$.
    Otherwise, since all the edges incident on $w$ have distinct colors, there exists such an edge $f$ such that $\chi(f) > \deg(w) + 2\beta (1 + \epsilon) \tilde \alpha_{\mathcal L(f)} + 1$, which contradicts the fact that $\chi$ was a good coloring of the graph before the deletion of $e$. It follows that $|\Gamma_w \cap \mathcal L^{-1}(j)| \leq 1$, so
    $$|\Gamma_w | = \sum_{j \in [L]} |\Gamma_w \cap \mathcal L^{-1}(j)| \leq L$$
    and hence $|\Gamma_u \cup \Gamma_v| \leq 2L$. To bound the number of edges that changed levels in at least one of the decompositions in the decomposition system, recall (see Proposition~\ref{prop:dynamic-decomp-sys}) that the amortized recourse of the algorithm that maintains the decomposition system is $O(L^2/\epsilon)$. It follows that the amortized number of such edges is $O(L^2/ \epsilon)$. We have that
    $$\frac{1}{T} \cdot \sum_{t \in [T]} |S^{(t)}| =  O \left(\frac{L^2}{\epsilon}\right) + 2L = O \left(\frac{\log^2 n}{\epsilon^3} \right).$$ 
\end{proof}

\bibliography{bibl.bib}

\section*{Acknowledgements}

Shay Solomon is funded by the European Union (ERC, DynOpt, 101043159).
Views and opinions expressed are however those of the author(s) only and do not necessarily reflect those of the European Union or the European Research Council. Neither the European Union nor the granting authority can be held responsible for them.
Shay Solomon and Nadav Panski are supported by the Israel Science Foundation (ISF) grant No.1991/1.
Shay Solomon is also supported by a grant from the United States-Israel Binational Science Foundation (BSF), Jerusalem, Israel, and the United States National Science Foundation (NSF).

\appendix


\section{The Sequential Static Algorithm}
\label{sec:static}

We begin with a simple static algorithm that computes a $\Delta + O(\alpha)$ edge coloring of a graph $G=(V,E)$.
We define $\Phi_{G'}(u,v)$ to be the minimum degree of the endpoints of the edge $(u,v)$ in the graph $G'$,
namely $\Phi_{G'}(u,v) = \min \{\deg_{G'}(u),\deg_{G'}(v)\}$.
Algorithm~\ref{alg:static} gives the pseudo-code for our static algorithm,
\textsc{CleverGreedy}.

\begin{algorithm}
    \SetAlgoLined
    \DontPrintSemicolon
    \KwIn{A graph $G = (V,E)$}
    \KwOut{An edge coloring $\chi$ of $G$}
    $\triangleright$ \texttt{PHASE I}\;
    $G_{1} \leftarrow G$\;
    \For{$i = 1,\dots ,m$} {
        $e_{i} \leftarrow \arg \min_{e \in E(G_i)} \Phi_{G_i}(e)$\;
        $G_{i+1} \leftarrow G_{i} - e_{i}$\;
    }
    $\triangleright$ \texttt{PHASE II}\;
    $\chi(e) \leftarrow \perp$ for all $e \in E$\;
    \For{$i = m,\dots ,1$} {
        Let $\chi(e_i)$ be any element in $[\deg_{G_i}(u_i) + \deg_{G_i}(v_i) - 1] \setminus \chi(N_{G_{i+1}}(e_i))$ where $e_i = (u_i,v_i)$\;
    }
    \Return{$\chi$}
    \caption{\textsc{CleverGreedy}$(G)$}
    \label{alg:static}
\end{algorithm}

\noindent The following theorem, which we prove next, summarizes the behavior of Algorithm \textsc{CleverGreedy}.

\begin{theorem}
    Algorithm \textsc{CleverGreedy} (Algorithm~\ref{alg:static}) is deterministic and, given a graph $G = (V,E)$ with maximum degree $\Delta$ and arboricity $\alpha$ as input, returns a $(\Delta + 2\alpha - 1)$-edge coloring of $G$, and can be implemented to run in $O(m \log^2 \Delta)$ time.
\end{theorem}

\subsection{Analysis of Algorithm \textnormal{\textsc{CleverGreedy}}}

Suppose we run Algorithm \textnormal{\textsc{CleverGreedy}} (Algorithm~\ref{alg:static}) on a
graph $G=(V,E)$ with maximum degree $\Delta$ and arboricity $\alpha$. 

We first show that Algorithm~\ref{alg:static} always returns a $(\Delta + 2\alpha - 1)$-edge coloring.

\begin{lemma}\label{lem:static}
    For all $i \in [m]$, we have that:
    \begin{enumerate}
        \item $[\deg_{G_i}(u_i) + \deg_{G_i}(v_i) - 1] \setminus \chi(N_{G_{i+1}}(e_i)) \neq \varnothing$, and
        \item $\deg_{G_i}(u_i) + \deg_{G_i}(v_i) - 1 \leq \Delta + 2\alpha - 1$.
    \end{enumerate}
\end{lemma}

\begin{proof}
 For (1), first note that
 $$|\chi(N_{G_{i+1}}(e_i))| \leq |N_{G_{i+1}}(e_i)| = \deg_{G_{i+1}}(u_i) + \deg_{G_{i+1}}(v_i) = \deg_{G_{i}}(u_i) + \deg_{G_{i}}(v_i) -2.$$
    Hence, $|[\deg_{G_i}(u_i) + \deg_{G_i}(v_i) - 1] \setminus \chi(N_{G_{i+1}}(e_i))| \geq 1$ and the claim follows. For (2), it is sufficient to show that $\Phi_{G_i}(e_{i}) \leq 2\alpha$, since this implies that one of the endpoints of $e_i$ has degree at most $2\alpha$.
    Let $G'_i$ denote the graph obtained by removing all isolated nodes from $G_i$.
    Since $G'_i$ is a subgraph of $G$, it has arboricity at most $\alpha$. In particular, $G'_i$ has average degree at most $2\alpha$, and so some node $u$ in $G'_i$ has $0 < \deg_{G'_i}(u) \leq 2\alpha$. Let $e$ be some edge in $G_i$ incident on $u$, then we have that $\Phi_{G_i}(e) \leq 2\alpha$, and so $\Phi_{G_i}(e_{i}) \leq \Phi_{G_i}(e) \leq 2\alpha$.
\end{proof}

\begin{corollary}
Algorithm \ref{alg:static} returns a $(\Delta + 2\alpha - 1)$-edge coloring of the graph.
\end{corollary}

\begin{proof}
    By Lemma \ref{lem:static}, we know that every edge $e$ in the graph receives a color $\chi(e) \in [\Delta + 2\alpha - 1]$ that is distinct from all of the colors assigned to adjacent edges.
\end{proof}

\noindent In Appendix~\ref{app:static runtime}, we give the full proof of the following lemma using the data structures presented in Appendix~\ref{sec:coloring data structs}.

\begin{lemma}\label{lem:static runtime}
    Algorithm \ref{alg:static} can be implemented to run in $O(m \log^2 \Delta)$ time.
\end{lemma}

\begin{proof}[Proof (Sketch).] Phase 1 of Algorithm~\ref{alg:static} can be implemented to run in $O(m)$ time by creating lists of the nodes with degree $d$ for each $d \in \{0,\dots,n-1\}$. We can then construct the sequence $e_1,\dots,e_m$ by repeatedly identifying a node with the smallest (non-zero) degree, removing all its incident edges, and updating the lists. It can be shown that this process terminates in $O(m)$ time.
Phase 2 of Algorithm~\ref{alg:static} can be implemented to run in $O(m \log^2 \Delta)$ time. This can be done by using an extension of the binary search data structure of \cite{BhattacharyaCHN18} in order to identify a color that is available to an edge in $O(\log^2 \Delta)$ time.
\end{proof}


\section{Data Structures}\label{sec:data structs}

In this appendix, we present the deterministic data structures used by our algorithm. We use two different data structures in order to implement our algorithm. The first is a data structure that we use to efficiently find new colors and is a variant of the binary search data structure of \cite{BhattacharyaCHN18}. The second is the data structure of \cite{BhattacharyaHNT15} that maintains the decompositions used by our algorithm, which we use as a black box.

\subsection{The Binary Search Data Structure}\label{sec:coloring data structs}

We now show how to implement a deterministic data structure that can dynamically maintain a collection of sets $\mathcal S = \{S_i\}_{i \in [n]}$, where each $S_i \subseteq \mathbb N$, and supports the following update and query operations.

\medskip
\noindent \textbf{Update:} Insert/delete a number $x \in \mathbb N$ from some set $S_i \in \mathcal S$.

\medskip
\noindent \textbf{Queries:} The data structure can answer two types of queries, as described below.

\begin{itemize}
    \item \textsc{Membership-Query}$(x, S_i)$: The input to this query is an integer $x$ and a set $S_i \in \mathcal S$. In response, the data structure outputs YES if $x \in S_i$ and NO otherwise.
    \item \textsc{New-Element}$(S_i, S_j)$: The input to this query is a pair of sets $S_i, S_j \in \mathcal S$. In response, the data structure outputs any number in $[|S_i| + |S_j| + 1] \setminus (S_i \cup S_j)$.
\end{itemize}

\noindent
Let $\tau$ be the largest integer across all of the sets $S_i$ in $\mathcal S$.
We now implement this data structure so that the \textsc{New-Element} query runs in $O(\log^2 \tau)$ time and the rest of the operations run in $O(\log \tau)$ time. We remark that the data structure does not need advanced knowledge of $\tau$ in order for these guarantees to hold, and that $\tau$ can change dynamically.

\medskip
\noindent \textbf{Implementing the Sets $S_i \in \mathcal S$.} We implement each set $S_i \in \mathcal S$ as a balanced binary search tree. This allows us to insert/delete elements from $S_i$ in $O(\log |S_i|)$ time. Furthermore, it also allows us to answer membership queries for the set $S_i$ in $O(\log |S_i|)$ time. After an update to the set $S_i$, the corresponding tree may need to be balanced, but this can be done in $O(\log |S_i|)$ time using standard implementations such as an AVL tree. We also implement these balanced binary trees so that each node within the tree explicitly maintains the size of its subtree. This can also be maintained efficiently since there are at most $O(\log |S_i|)$ nodes in the tree whose subtrees change as the result of an update and the proceeding balancing operation.

\medskip
\noindent \textbf{Implementing the \textnormal{\textsc{New-Element}} Query.} Given sets of positive integers $S_i$ and $S_j$, it's easy to see that $[|S_i| + |S_j| + 1] \setminus (S_i \cup S_j)$ is non-empty. In order to implement this query, we first show how to implement the following query and then give a reduction:
\begin{itemize}
    \item \textsc{Set-Intersection}$(S_i, k_1, k_2)$: The input to this query is a set $S_i \in \mathcal S$, and positive integers $k_1$ and $k_2$ such that $k_1 \leq k_2$. In response, the data structure returns the size of the set $S_i \cap [k_1,k_2]$.
\end{itemize}
\begin{claim}
The query \textsc{Set-Intersection}$(S_i, k_1, k_2)$ can be implemented to run in $O(\log |S_i|)$ time.
\end{claim}

\begin{proof}
    Suppose that given some set $S_i$, implemented as a balanced binary tree as described in the preceding section, and some positive integer $k$, we can find $|S_i \cap [k]|$ in $O(\log |S_i|)$ time. Then, by noting that $|S_i \cap [k_1,k_2]| = |S_i \cap [k_2]| - |S_i \cap [k_1 - 1]|$, it follows that \textsc{Set-Intersection} queries can be answered in $O(\log |S_i|)$ time. We now show how to compute $|S_i \cap [k]|$ in $O(\log |S_i|)$ time.

    We can implement this query recursively. Given some node $x$ in the tree representation of $S_i$, let $\texttt{val}(x)$ denote the value stored at the node $x$, and let
    $Q_k(x)$ denote the number of nodes in the subtree of $x$ (including $x$) whose values are at most $k$.
    Now, let $y$ and $z$ be the left and right children of $x$ respectively. If $\texttt{val}(x) > k$, then none of the nodes stored in the subtree of $z$ have values at most $k$, so $Q_k(x) = Q_k(y)$. Otherwise, if $\texttt{val}(x) \leq k$, then all of the nodes stored in the subtree of $y$ have values at most $k$, so $Q_k(x) = \texttt{size}(y) + 1 + Q_k(z)$, where $\texttt{size}(y)$ denotes the size of the subtree rooted at $y$ (and is explicitly maintained at $y$). If $x$ does not have a left or right child, we just omit the corresponding terms. Since the tree has depth $O(\log |S_i|)$, this can be done in $O(\log |S_i|)$ time.
\end{proof}

\noindent
We now show how the \textsc{Set-Intersection}$(S_i, k_1, k_2)$ query can be used to implement the \textsc{New-Element}$(S_i, S_j)$ query with only $O(\log \tau)$ overhead.

\begin{lemma}
    The query \textsc{New-Element}$(S_i, S_j)$ can be implemented to run in $O(\log^2 (|S_i| + |S_j|))$ time.
\end{lemma}

\begin{proof}
    We do this by performing a binary search that we implement using the \textsc{Set-Intersection} query. Given sets $S_i$ and $S_j$, first set $k_1 = 1$ and $k_2 = |S_i| + |S_j| + 1$, and note that
    \begin{equation*}\label{eq:invariant}
        |S_i \cap [k_1,k_2]| + |S_j \cap [k_1,k_2]| < k_2 - k_1 + 1.
    \end{equation*}
    We can then perform a binary search by splitting the interval $[k_1, k_2]$ into disjoint intervals $I_\ell$ and $I_r$ such that $[k_1, k_2] = I_\ell \cup I_r$ and $|I_\ell| = |I_r| \pm 1$ and noting that, for some $I \in \{ I_\ell, I_r \}$, we have that $|S_i \cap I| + |S_j \cap I| < |I|.$ We can find such an interval $I$ using the \textsc{Set-Intersection} query in $O(\log |S_i| + \log |S_j|)$ time. By repeating this for $O(\log(|S_i| + |S_j|))$ many iterations, the size of the interval becomes 1 and we find a number $x$ such that $|S_i \cap \{x\}| + |S_j \cap \{x\}| < 1$ and hence $x \notin S_i \cup S_j$. It follows that this process takes $O(\log^2 (|S_i| + |S_j|))$ time in total. 
\end{proof}

\subsection{The Key Data Structure}

We now describe our key data structure that we use to implement our dynamic algorithm. Our key data structure is an extension of the data structure described in Proposition~\ref{prop:dynamic-decomp-sys} that also incorporates the binary search data structure in order to maintain a coloring. We restate Proposition~\ref{prop:dynamic-decomp-sys} for convenience.

\begin{proposition}
For any constant $\beta \geq 2 + 3 \epsilon$, there is a deterministic fully-dynamic algorithm that maintains a $(\beta, (d_j)_{j \in [K]}, L)$-decomposition system of a graph $G = (V,E)$ with amortized update time and amortized recourse $O(KL/\epsilon)$.
\end{proposition}

\noindent
It follows from the details of the data structure in \cite{BhattacharyaHNT15} that this data structure maintains the sets of edges $N_{j}^+(u)$ and the levels of nodes $\ell_j(u)$ \textit{explicitly} for each $u \in V$ and $j \in [K]$. In order to extend this data structure and use it to maintain a coloring $\chi$ of the graph $G$ as it undergoes updates, we also maintain the following auxiliary data structures.

\begin{itemize}
    \item For all $e \in E$, $\textsc{Color}(e)$ keeps track of the color currently assigned to $e$, i.e. $\chi(e)$. We implement this using balanced search trees so it takes $O(\log n)$ time to get/set the color of an edge $e$.
    
    \item For all $j \in [K]$, $u \in V$, $\textsc{Colors}_j^+(u)$ is the set of all colors assigned to the edges in $N^+_j(u)$, i.e. $\chi(N_{j}^+(u))$. We implement these sets using balanced search trees in the same way as the sets in Appendix~\ref{sec:coloring data structs}, allowing us to perform \textsc{New-Element} queries on these sets in $O(\log^2 \Delta)$ time and insert/delete elements from the set in $O(\log \Delta)$ time.
    
    \item For all $u \in V$, $\textsc{Colors}(u)$ is the set of all colors assigned to the edges incident on $u$, i.e. $\chi(N(u))$. Like the previous sets, we implement these sets using balanced search trees in the same way as the sets in Appendix~\ref{sec:coloring data structs}, with the additional modification that we add a pointer from each color $c \in \textsc{Colors}(u)$ to the edge incident on $u$ that is assigned color $c$, allowing us to search for such edges in $O(\log \Delta)$ time.
\end{itemize}

\noindent We now argue that we can still efficiently maintain this modified data structure, which we call $\mathcal D$, in the dynamic setting with only $\tilde O(1)$ overhead.

\medskip
\noindent \textbf{Updating the Coloring $\chi$.} Suppose that we want to update the color $\chi(e)$ of an edge $e = (u,v)$ in $G$. Clearly, this does not require us to change the structure of the underlying decomposition system. $\textsc{Color}(e)$ can be updated in $O(\log n)$ time. $\textsc{Colors}(u)$ and $\textsc{Colors}(v)$ can be updated in $O(\log \Delta)$ time by adding and removing at most 1 color from each. Finally, for each $j \in [K]$, we can check $\ell_j(u)$ and $\ell_j(v)$ in $O(1)$ time, and update $\textsc{Colors}_j^+(u)$ and $\textsc{Colors}_j^+(v)$ in $O(\log \Delta)$ time, taking $O(K \log \Delta)$ time in total. It follows that, after changing the color of an edge, the data structure can be updated in $O(\log n +  K\log \Delta)$ worst-case time.

\medskip
\noindent \textbf{Updating the Decomposition System.} Suppose that an edge insertion or deletion leads us to update the decomposition $(Z_{i,j})_{i,j}$. Since this does not change the underlying coloring, the data structures $\textsc{Color}$ and $\textsc{Colors}$ do not change. However, since the levels of nodes might change, $\textsc{Color}^+$ might change. We can observe that colors only need to be inserted/deleted from $\textsc{Color}_j^+(u)$ when edges are inserted/deleted from $N_{j}^+(u)$, which the data structure in Proposition~\ref{prop:dynamic-decomp-sys} maintains explicitly. Hence, we make the modification that every time an edge $e$ is inserted/deleted from $N_{j}^+(u)$ we also insert/delete $\chi(e)$ from $\textsc{Color}_j^+(u)$. Since it takes $O(\log \Delta)$ time to insert/delete element from $\textsc{Color}_j^+(u)$, this leads to at most $O(\log \Delta)$ overhead. We get the following lemma about the behavior of our key data structure.

\begin{lemma}\label{lem:data structure 2}
    For any constant $\beta \geq 2 + 3 \epsilon$, the data structure $\mathcal D$ maintains a $(\beta, (d_j)_{j \in [K]}, L)$-decomposition of a graph $G = (V,E)$ with amortized update time $O(KL \log\Delta /\epsilon)$ and amortized recourse $O(KL /\epsilon)$, and can update the colors of edges in $O(\log n + K\log \Delta)$ worst-case time.
\end{lemma}

\noindent Note that, for the specific parameters used by our dynamic algorithm, this corresponds to an amortized update time of $O(\log^2 n \log\Delta /\epsilon^3)$, an amortized recourse of $O(\log^2 n /\epsilon^3)$, and $O(\log n \log\Delta /\epsilon)$ time to update colors.

\subsection{Implementing the Dynamic Algorithm}\label{sec:implementation}

Suppose that our algorithm handles a sequence of $T$ edge insertions and deletions for an initially empty dynamic graph $G$. For $t \in [T]$, let $S^{(t)}$ denote the set of edges that are uncolored by our algorithm during the $t^{th}$ update before calling $\textsc{ExtendColoring}$ on the set $S^{(t)}$. We now show that our algorithm has an amortized update time of $O(\log^5 n \log \Delta / \epsilon^6)$.

\begin{lemma}\label{lem: total insert / delete time}
    The total time spent handling calls to \textsc{Insert} (Algorithm~\ref{alg:update 2}) and \textsc{Delete} (Algorithm~\ref{alg:deletion 2}), excluding the time taken to handle calls to \textsc{ExtendColoring}, is $T \cdot O(\log^2 n \log \Delta / \epsilon^3) + O(\log n \log \Delta / \epsilon) \cdot \sum_{t \in [T]} |S^{(t)}|$.
\end{lemma}

\begin{proof}
    We first remark that we do not need to explicitly store the graph $G$, since the data structure $\mathcal D$ is sufficient to implement the algorithm. In the case of a deletion, we can initially construct the set $S$ in $O(\log n \log \Delta / \epsilon)$ time. To see why, recall that, in the proof of Lemma \ref{lem:full algorithm recourse}, we show that for any $w \in e$, $|\Gamma_w| \leq L$, and we know which $L$ colors $c_1,...,c_L$ can be taken by the edges in this set. Hence, for each such color $c_j$, we can just check whether there is an edge incident on $w$ with color $c_j$, and if there is such an edge we can check if it is bad and add it to $S$, which be implemented as a linked list. Overall, this can be done in $O(L\log \Delta)$ time.
    
    By Lemma~\ref{lem:data structure 2}, the total time spent updating the decomposition system across all the updates is $T \cdot O(\log^2n \log \Delta / \epsilon^3)$. Since the data structure in Lemma~\ref{lem:data structure 2} maintains the levels of edges explicitly, we can identify the set $S'$ of edges that change levels in some layer during an update with $O(1)$ overhead while updating the decomposition system. We can then take the set $S$ and uncolor all of these edges in $|S| \cdot O(\log n \log \Delta / \epsilon)$. Summing over all updates, we get the desired bound.
\end{proof}

\begin{lemma}\label{lem: extend coloring time}
    The total time spent handling calls to algorithm $\textsc{ExtendColoring}$ (Algorithm~\ref{alg:extend 2}) is $O( \log^3 n \log \Delta / \epsilon^3) \cdot \sum_{t \in [T]}|S^{(t)}|$.
\end{lemma}

\begin{proof}
    Suppose that we run $\textsc{ExtendColoring}$ on input $S$. By Lemma~\ref{lem:full dynamic anal 1}, we know that the while loop in Algorithm \ref{alg:extend 2} runs for at most $L^2|S|$ iterations. We now show that each iteration takes $O(\log n \log \Delta/\epsilon)$ time. We first recall that the data structure $\mathcal D$ maintains the levels of nodes and edges explicitly and allows us to update the colors of edges in $O(\log n \log \Delta / \epsilon)$ time, as well as explicitly maintaining the orientation of the edges in each layer. Furthermore, since we implement the sets $\textsc{Colors}_j^+$ and $\textsc{Colors}$ so that they support $\textsc{New-Element}$ queries, we can find the color $c$ in $O(\log^2 \Delta)$ time. Hence, each line in the while loop can be implemented to run in $O(\log n \log \Delta/ \epsilon)$ time. It follows that $\textsc{ExtendColoring}$ runs in $|S| \cdot O( \log^3 n \log \Delta / \epsilon^3)$ time. Summing over all updates, we get the desired bound.
\end{proof}

\begin{lemma}
    The dynamic algorithm has an amortized update time of $O(\log^5 n \log \Delta / \epsilon^6)$.
\end{lemma}

\begin{proof}
    It follows from Lemmas~\ref{lem: total insert / delete time} and \ref{lem: extend coloring time} that the total time our dynamic algorithm spends handling updates is
    $$T \cdot O \left(\frac{\log^2 n \log \Delta}{\epsilon^3} \right) + O\left(\frac{\log n \log \Delta}{\epsilon} \right) \cdot \sum_{t \in [T]} |S^{(t)}| + O \left(\frac{\log^3 n \log \Delta}{\epsilon^3} \right) \cdot \sum_{t \in [T]}|S^{(t)}| $$
    $$ = T \cdot O\left(\frac{\log^2 n \log \Delta}{\epsilon^3}\right) + O\left(\frac{\log^3 n \log \Delta}{ \epsilon^3}\right) \cdot \sum_{t \in [T]}|S^{(t)}|. $$
    It follows from the proof of Lemma~\ref{lem:full algorithm recourse} that $\sum_{t \in [T]}|S^{(t)}| = T \cdot O(\log^2 n / \epsilon^3)$, and hence the total update time of our algorithm is $T \cdot O(\log^5 n \log \Delta / \epsilon^6)$ and the lemma follows.
\end{proof}

\section{Deferred Proofs}\label{app:deferred proofs}

\subsection{Proof of Lemma~\ref{lem:static runtime}}\label{app:static runtime}



\begin{lemma}\label{lem:static phase 1 implementation new}
    Phase 1 of Algorithm \ref{alg:static} can be implemented to run in $O(m)$ time.
\end{lemma}

\begin{proof}
    We begin by creating $n$ lists of nodes $L_0,\dots,L_{n-1}$, one corresponding to each possible degree of a node in the graph $G$, and place each node $u$ into the list $L_{\deg(u)}$. We also create a pointer from each node $u$ to its position in list $L_{\deg(u)}$. We now show how to use these lists in order to construct the sequence $e_1,\dots,e_m$. 
    Suppose that we have already found $e_1,\dots,e_{i-1}$, and that each node $u$ is contained in the list $L_{\deg_{G_i}(u)}$, i.e. is contained in the list corresponding to its degree in the graph $G_i$.
    Suppose that the node with the smallest non-zero degree in $G_i$ is $u$, and has degree $\delta$. Then we can find this node in $O(\delta)$ time by scanning through the lists $L_1,L_2,\dots$ until we find the first non-empty list $L_\delta$. We then take all of the $\delta$ edges incident on $u$ in $G_i$, assign them to be the edges $e_{i},\dots,e_{i+\delta-1}$, remove them from $G_i$ to obtain $G_{i+\delta-1}$, and update the positions of the nodes in the lists to correspond to their degrees in $G_{i+\delta-1}$. This can be done in $O(\delta)$ time and adds $\delta$ edges to the list. By repeating this process until all the nodes have degree $0$, we can construct the list $e_1,\dots,e_m$ in $O(m)$ time.
\end{proof}

\begin{lemma}\label{lem:static phase 2 implementation}
    Phase 2 of Algorithm \ref{alg:static} can be implemented to run in $O(m \log^2 \Delta)$ time.
\end{lemma}

\begin{proof}
    For each node $u \in V$, we can dynamically maintain the set $\chi (N_{G_{i+1}}(u))$ using the binary search data structure (see Appendix~\ref{sec:coloring data structs}) as edges are colored during phase 2 of Algorithm \ref{alg:static}. Given nodes $u, v \in V$, we can then use this data structure to perform a \textsc{New-Element} query on the sets $\chi (N_{G_{i+1}}(u))$ and $\chi (N_{G_{i+1}}(v))$ in order to find a color from the set $[\deg_{G_{i+1}}(u_i) + \deg_{G_{i+1}}(v_i) + 1] \setminus \chi(N_{G_{i+1}}(e_i))$ in $O(\log^2(\deg_{G_{i+1}}(u_i) + \deg_{G_{i+1}}(v_i))) \leq O(\log^2\Delta)$ time, which implements one iteration of the loop in phase 2. Since the data structure can be updated in $O(\log \Delta)$ time, the lemma follows.
\end{proof}

\noindent Lemma~\ref{lem:static runtime} follows immediately from Lemmas \ref{lem:static phase 1 implementation new} and \ref{lem:static phase 2 implementation}.

\subsection{Proof of Lemma~\ref{lem:warmup dynamic update time}}\label{app:warmup dynamic}

In order to implement the warmup dynamic algorithm, we can use the data structure $\mathcal D$ from Lemma~\ref{lem:data structure 2} with $K=1$ since we only need to maintain a single graph decomposition. With our parameters, the data structure $\mathcal D$ has an amortized update time of $O(\log n \log \Delta / \epsilon^2)$ and can update the colors of edges in $O(\log n)$ worst-case time.

\begin{lemma}\label{lem:dynamic update time 1}
    Given any sequence of updates, the amortized time taken to update the $(\beta, 2(1 + \epsilon)\alpha, L)$-decomposition $(Z_i)_i$ of $G$ is $O(\log n \log \Delta /\epsilon^2)$.
\end{lemma}

\begin{proof}
    This follows immediately from the fact that the data structure $\mathcal D$ has an amortized update time of $O(L \log \Delta /\epsilon)$.
\end{proof}

\begin{lemma}\label{lem:dynamic update time 2}
    Given some uncolored edge $e$ and a $(\beta, 2(1 + \epsilon)\alpha, L)$-decomposition $(Z_i)_i$ of $G$, the algorithm $\textsc{ExtendColoring}(e, (Z_i)_i)$ (Algorithm~\ref{alg:extend}) can be implemented to run in $O((\log^2 \Delta + \log n)\log n / \epsilon)$ worst-case time.
\end{lemma}

\begin{proof}
    By the proof of Lemma~\ref{lem:dynamic anal 1}, we know that the while loop in Algorithm \ref{alg:extend} runs for at most $L$ iterations. We now show that each iteration runs in $O(\log^2 \Delta + \log n)$ time. Recall that we can update the colors of edges in $O(\log n)$ time. Furthermore, the color $c$ can be found in $O(\log^2 \Delta)$ time since we implement the sets $\textsc{Colors}^+(u)$ and $\textsc{Colors}(v)$ so that they support $\textsc{New-Element}$ queries which run in $O(\log^2 \Delta)$ time. Hence, each line in the while loop can be implemented to run in $O(\log^2 \Delta + \log n)$ time.
\end{proof}

\noindent
The time taken to handle the deletion of an edge $e$ is the time taken to delete $e$ from the graph, uncolor the edge $e$ in our data structure, and update the decomposition of the graph. Since we can delete and uncolor the edge in $O(\log n)$ time, by Lemma~\ref{lem:dynamic update time 1}, this can be done in $O(\log n \log \Delta / \epsilon^2)$ amortized time. If we are handling an insertion, we also need to handle a call to $\textsc{ExtendColoring}(e, (Z_i)_i)$, which by Lemma~\ref{lem:dynamic update time 2} runs in time $O((\log^2 \Delta + \log n)\log n / \epsilon)$, on top of updating the graph and the decomposition. Hence, the amortized update time of the warmup algorithm is $O(\log n \log \Delta / \epsilon^2) + O((\log^2 \Delta + \log n)\log n / \epsilon)$, which is $O(\log^2 n \log \Delta / \epsilon^2)$.

\end{document}